\documentclass[conference]{IEEEtran}
\makeatletter
\def\ps@headings{
\def\@oddhead{\mbox{}\scriptsize\rightmark \hfil \thepage}
\def\@evenhead{\scriptsize\thepage \hfil \leftmark\mbox{}}
\def\@oddfoot{}
\def\@evenfoot{}}
\makeatother
\usepackage[pdftex]{graphicx}

\usepackage[cmex10]{amsmath}
\usepackage{amsfonts,amstext,amssymb}

\usepackage{url}

\newtheorem{rem}{Remark}
\newtheorem{thm}{Theorem}

\newtheorem{prop}{Proposition}
\def\R{\mathbb{R}}

\newcommand{\E}[1]{\mathbb{E}#1}

\usepackage{color}
\usepackage{threeparttable}
\usepackage{ifthen}
\usepackage[latin1]{inputenc}
\usepackage{epstopdf}
\usepackage{paralist}

\newenvironment{itemize*}
  {\begin{itemize}
    \setlength{\itemsep}{0pt}
    \setlength{\parskip}{0pt}}
  {\end{itemize}}

\newboolean{showcomments}
\setboolean{showcomments}{true}
\ifthenelse{\boolean{showcomments}}
{ \newcommand{\mynote}[3]{
    \fbox{\bfseries\sffamily\scriptsize#1}
    {\small$\blacktriangleright$\textsf{\emph{\color{#3}{#2}}}$\blacktriangleleft$}}}
{ \newcommand{\mynote}[3]{}}

\begin{document}
\title{Can P2P Networks be Super-Scalable? \vspace{-.3cm}}

\author{\IEEEauthorblockN{Fran{\c c}ois Baccelli}
\IEEEauthorblockA{UT Austin \& INRIA -- \'ENS\\
USA
}
\and
\IEEEauthorblockN{Fabien Mathieu}
\IEEEauthorblockA{INRIA -- University Paris 7\\
France
}
\and
\IEEEauthorblockN{Ilkka Norros}
\IEEEauthorblockA{VTT\\
Finland
}
\and
\IEEEauthorblockN{R\'emi Varloot}
\IEEEauthorblockA{INRIA\\
France
}
}

\maketitle

\IEEEpeerreviewmaketitle

\begin{abstract}
We propose a new model for peer-to-peer networking which takes
the network bottlenecks into account beyond the access. This
model can cope with key features of P2P networking
like degree or locality constraints together with the fact that distant peers
often have a smaller rate than nearby peers.

Using a network model based on rate functions, we give a closed form expression of peers download performance in the system's fluid limit, as well as approximations for the other cases. Our results show the existence of realistic settings for which the average download time is a decreasing function of the load, a phenomenon that we call super-scalability. 

\end{abstract}

\section{Introduction}

The Peer-to-Peer (P2P) paradigm has been widely used to quickly deploy low-cost, scalable, decentralized architectures. For instance, the success of BitTorrent~\cite{bittorrent} has shown that file-sharing can be provided with full scalability. Although many other architectures currently compete with P2P (dedicated Content Distribution Networks, Cloud-based solutions, \ldots),
P2P is still unchallenged with respect to its low-cost and scalability features, and remains a major actor in the field of content distribution.

Today, the main limitation for P2P content distribution is probably the access upload bandwidth, as even high-speed Internet access connections are
often asymmetric with a relatively low uplink capacity. Therefore most P2P content distribution performance 
models assume a relatively
low access bandwidth as the main performance bottleneck. However, in a near
future the deployment of very high speed access (e.g. FTTH) will challenge the
justification of this assumption. 
This raises the need of new P2P models that describe what
happens when the access is not necessarily the main/only bottleneck
and that allow one to better understand the fundamental
limitations of P2P.

\vspace{-.15cm}
\subsection{Contributions}
\vspace{-.15cm}

{\bf {A new model.}} The first contribution of the present paper is
the model presented in Section~\ref{sec:model}, which
features the following two key ingredients:
1) a spatial component thanks to which the topology of
the peer locations is used to determine their interactions
2) a networking component allowing one to represent the
actual exchange throughput between peers.

{\bf {A promising form of scalability.}} 
In most P2P bandwidth models, the upload/download capacity is the bottleneck
determining the exchange throughput obtained by peers~\cite{Veciana03fairness,qiusrikant04,benbadis08playing}. This creates \emph{scalability}, where the download latency remains constant when the system load increases. Our new model exhibits a stronger form of scalability, which we call \emph{super-scalability}, where the service latency actually {\em decreases} with the
system load.

We show in Sections \ref{sectoy} and \ref{sec:theory} that super-scalability is a consequence of network dynamics causing the service rate of a typical customer to increase with the load of the system.

{\bf {Conditions for super-scalability to hold.}} 
One may question the realism of such a model, as the underlying network obviously cannot sustain arbitrarily high rates. Section~\ref{sec:lim} combines our model with an abstract (physical) network model to determine the conditions for which our model makes sense and super-scalability occurs.

Another natural issue is data availability: bandwidth can be a bottleneck only if peers have something to transmit to each other. We address this issue in Section~\ref{sec:chunks}, where we study the impact of data availability on the effective download performance.

{\bf {The laws of super-scalability.}}
Starting from the basic model studied in Section~\ref{sec:theory}, we build in Section \ref{sec:extend} a Swiss Army Knife for handling many realistic variants: generic rate functions, auxiliary servers, seeding behavior of users, access bottleneck conditions\ldots The corresponding laws determine optimal tuning of the parameters of the P2P algorithms e.g. peering degree, transport protocol or seeding times.

\vspace{-.15cm}
\subsection{Related Work}
\vspace{-.15cm}

Our main scenario 
is inspired by a BitTorrent-like file-sharing protocol. In BitTorrent \cite{bittorrent}, a file is segmented
into small chunks and each downloader (called \emph{leecher}) exchanges
chunks with its neighbors in a peer-to-peer overlay network. A peer may continue to distribute chunks after it has completed its own download (it is then called a \emph{seeder}).
Here is a short summary of what is kown on this scenario.

{\bf {Bandwidth-centered modeling.}} Some studies have  analyzed the effectiveness of P2P file-sharing with a simple dynamic system model of peer arrival, focusing on the
performance under the assumption that the access bandwidth is the main bottleneck~\cite{Veciana03fairness,qiusrikant04,benbadis08playing}. While the present paper focuses on a similar bandwidth-centered approach, it introduces a richer family of peer interaction models.

{\bf {Chunk availability.}} Another potential bottleneck is chunk availability. The worst possible case is the ``missing piece syndrome''~\cite{mathieu:missing}, where one chunk keeps existing in only a few copies (or none!) and the peer population can grow unboundedly while trying to get that chunk. The syndrome may happen for some scenarios~\cite{hajekzhu10,zhuhajekarx11}, but it can be avoided by using more or less sophisticated download policies, at the cost of somewhat increased download times, see \cite{hajekzhu10,zhuhajekarx11,reittuiwsos09,norreiei11,oguzanantnor-arxiv}.
Also note that
\cite{massoulievojnovic08} proposed an elegantly abstracted stochastic
chunk-level model of uncoordinated file-sharing. The results in \cite{massoulievojnovic08}
indicate that if the system has high input rate and starts with a 
large and sufficiently balanced population of chunks, it may perform for a long time without missing chunk even if there is no seeder.

In this paper, we assume that missing chunk issues are avoided by some mechanism (like getting the \emph{locally rarest} chunk with high priority), so the impact of chunk on performance is reasonable. Nevertheless, we estimate this impact through a very simple chunk-level modeling, inspired by the ones proposed in \cite{qiusrikant04} and \cite{massoulievojnovic08}.

{\bf{Spatially-dependent rate.}} While a large number of studies consider the case of heterogeneous rates, to the best of our knowledge, none considers a system where the transfer speeds depend on pair-wise distances but not on the nodes as such. There are some earlier papers considering P2P systems in a spatial framework (for instance, \cite{pannet-susitaival-06}), but they do not assume that distance has some effect on transfer speed. 
Our paper seems to be the first where a peer's
downloading rate is a function of its distances to other peers.

\vspace{-.1cm}
\section{Super--scalability Toy Example}
\label{sectoy}
\vspace{-.1cm}

Before getting into the core of the paper, consider a system in steady state where 
peers arrive with some arrival intensity $\lambda$,
download some file of size $F$
and leave the system as soon as their own download is completed.
We neglect here geometry as well
as chunk availability issues. By the latter we mean that a peer has always a chunk to
provide for another, unfinished peer.

Suppose that the access upload bandwidth is
the main bottleneck. If $U$ is the typical upload bandwidth
of a peer, then it makes sense to assume that $U$ is also the
typical download throughput experienced by each peer. In particular,
in the steady state (if any), the mean latency $W$ and the average
number of peers $N$ should be such that
\vspace{-.2cm}
\begin{equation}
	W=\frac{F}{U}\text{ and }\
	N=\lambda W=\frac{\lambda F}{U}\text{ (Little's Law).}
\label{eq:bottleneck-1.0}
\vspace{-.15cm}
\end{equation}
Although very simple, \eqref{eq:bottleneck-1.0} contains a core property of
standard P2P systems: the mean latency is independent of the arrival rate.
This is the \emph{scalability} property, one of the main
motivations for using P2P.

Now, imagine a complete shift of the bottleneck paradigm.
Let the main resource bottleneck be the (logical, directed)
links between nodes instead of the nodes themselves. We should
then consider the typical bandwidth $U$ \emph{from one peer to
another} as the key limitation. If each peer is connected to every
other one (the interaction graph is complete at any time),
then Equation \eqref{eq:bottleneck-1.0}
should be replaced by
$W=\frac{F}{(N-1)U}$ and $N=\lambda W$, which leads to
\vspace{-.15cm}
\begin{equation*}
N=\sqrt{\frac{\lambda F}{U}+\frac{1}{4}}+\frac{1}{2}\text{ and }
W=\sqrt{\frac{F}{\lambda U}+\left(\frac{1}{2\lambda}\right)^2}+\frac{1}{2\lambda}\text{.}
\label{eq:bottleneck-2.1}
\vspace{-.15cm}
\end{equation*}
For $\frac{\lambda F}{U}\gg1$, this can be approximated by
\begin{equation}
N\approx\sqrt{\frac{\lambda F}{U}}\text{ and }\\
W\approx\sqrt{\frac{F}{\lambda U}}\text{.}\\
\label{eq:bottleneck-2.2}
\end{equation}
Now, the service time is inversely proportional
to the square root of the arrival intensity: this is \emph{super-scalability}.

\begin{rem} \emph{In fact, the real solution is a little bit more complex than that due to size fluctuations that have not been taken into account here. A more rigorous description of the toy model is available in \cite{baccelli:inria-00615523}.}
\end{rem}

In this toy example, the central reason for super-scalability
is rather obvious: the number of edges in a complete graph
is of the order of the square of the number of nodes,
and so is the overall service capacity.

The main question addressed in the present paper is
to better understand the fundamental limitations of P2P systems
and in particular to check whether super-scalability can possibly
hold in future, network-limited, P2P systems, where the throughput
between peers will be determined by transport protocols and  
network resource limitations rather than the upload capacity alone.
This requires the definition of a new model allowing one
to capture the toy model idea while taking into account the limitations inherent to P2P overlays
as well as network capacity constraints.

\vspace{-.1cm}
\section{Network Limited P2P Systems}
\vspace{-.1cm}
\label{sec:model}
\def\R{\mathbb{R}}

The aim of this section is to define a basic model that tries to capture super-scalability, spatially dependent rates and P2P constraints. This model will be extended in the last sections of the paper.

\begin{table}
\centering
\caption{Notation for the Basic Model}
\vspace{-.15cm}
{
\begin{tabular}{|l|l|l|}
\hline
Name & Description & Units \\
\hline
\hline
$\lambda$ & Leecher arrival rate & $m^{-2}\cdot s^{-1}$\\
\hline
$C$ & Rate parameter & $bits\cdot s^{-1}\cdot m$ \\
\hline
$F$ & Mean file size & $bits$ \\
\hline
$R$ & Peering range & $m$\\
\hline
$W$ & Mean latency & $s$\\
\hline
$\mu$ & Mean rate & $bits\cdot s^{-1}$\\
\hline
$\beta$ & Peer density & $m^{-2}$\\
\hline
\end{tabular}
}
\vspace{-.5cm}
\end{table}

{\bf{Spatial domain.}} Our peers live in a domain $D$ equipped with a distance $d$. The meaning of $d$ can be manyfold: physical distance; latency-based pseudo-distance~\cite{flv08}; $D$ can even be some representation of peer categories, the position of a peer representing its own centers of interest. The main point is that we assume that the rate between two peers depends on their distance in $D$.
For simplicity, we focus on a basic model where $D$ is an arbitrarily large torus that approximates the Euclidean plane $\R^2$, but there is no basic difficulty in extending this framework to other topologies better suited to model networks, like a hyperbolic space~\cite{hyperbolic}. Distances in $D$ are expressed in meters, regardless of the actual meaning of $D$.

{\bf{Arrival rate.}} We assume that new peers arrive according to a
Poisson process with space-time intensity $\lambda$ ({\em
``Poisson rain''}). The parameter $\lambda$, expressed in $m^{-2}\cdot s^{-1}$, 
describes the birth rate of peers: the number of peer that arrive in
a domain of surface $A$ (expressed in $m^2$) in an interval $[s,t]$ (in seconds) is a Poisson random variable with parameter $\lambda A(t-s)$.

{\bf{Data rate.}}
For our basic model, we assume that the transfer rate is determined by a congestion mechanism like TCP Reno. On the path between two peers, let $\vartheta$ denote the packet loss probability and $\mathrm {RTT}$ the round trip time.  Then the square root formula \cite{ott} stipulates that the rate obtained on this path is
$ \frac{\xi}{{\mathrm {RTT}} \sqrt{\vartheta}},$ with $ \xi= \sim1.309$. Assuming the RTT to be proportional to distance $r$ yields a transfer rate of the form
\vspace{-.1cm}
\begin{equation}
\label{eqr0}
f(r)=\frac{C}{r}\text{,}
\vspace{-.15cm}
\end{equation}
\noindent where $C$ is a rate parameter expressed in $bits\cdot s^{-1}\cdot m$.

We assume that the rates are additive, so that the total download rate of a peer $x$ is 
\vspace{-.15cm}
\begin{equation}
\label{eqdeathrate}
\mu(x)=\sum_{y\in N(x)}f(d(x,y))\text{,}
\vspace{-.15cm}
\end{equation}
\noindent where $N(x)$ is the set of neighbors of $x$ (in the overlay) and $d(x,y)$ the distance between $x$ and $y$.

We consider symmetric connections, because: the data rate function is symmetric; chunk availability may be neglected for proper parameters (see Section \ref{sec:chunks}); some tit-for-tat mechanisms may be at play to enforce some kind of reciprocity between peers. By symmetry, $\mu(x)$ is also the upload rate of a peer 
at $x$. In order for the access not to be a further limitation,
the access capacity of a peer at $x$ should exceed $\mu(x)$.
This is our default assumption here (access as a possible bottleneck is considered in
Section~\ref{sec:extend}).

The choice of a rate function given by \eqref{eqr0} is mainly for giving explicit results based on a simple distance-varying rate. Our results indeed apply for a wide range of rate functions (cf. Section \ref{sec:generalrate}).

{\bf{Data size.}} Each peer ${p}$ wants to get an amount $F_{p}>0$ of data. In the basic BitTorrent example where every peer wants to get the same file, $F_{p}$ would most naturally be modeled by a constant $F$ (the size of the file).
For the sake of mathematical tractability, in the analytical models, we follow the approach used by \cite{qiusrikant04} and assume that the $F_{p}$'s are independent and identically distributed random variables, with finite expectation $F=\E(F_{p})$.

{\bf Unaltruism.} When a peer has finished its download, it leaves the system immediately (instead of becoming a seeder).

{\bf {Connectivity limitation.}} The toy example assumes full mesh connectivity between peers, which is not a reasonable assumption. In practice, peers usually limit their neighborhood by using some \emph{overlay graph}. There are many ways to build an overlay, for instance by selecting only peers with sufficient qualities and/or by limiting their total number of neighbors. 
 In the basic model, we propose to define connectivity by a \emph{range} $R$: if $\Phi_t$ is the set of peers present at time $t$, then $N_t(x)=\{y\in \Phi_t, y\neq x, \text{s.t. } d(x,y)\leq R\}$.
The range can for instance originate from an ALTO-like connection management that prevents peers too far from one another to connect~\cite{bookchapter2009-Hossfeld}.
This constraint is even more meaningful in a wireless context, as it can represent the transmission range.

Other connectivity rules could be enforced, for instance random connectivity, but if the rate function decreases with the distance, it is only natural to enforce proximity in the overlay graph. Later in the paper (Section~\ref{sec:extend}), we propose another proximity-based variant where a constant number of closest peers is selected.

{\bf{Chunks.}} In order to focus on bandwidth aspects, the basic model follows the approach proposed by \cite{qiusrikant04}: we assume that the effect of chunk (un)availability between peers is that the download effectiveness is affected by some factor $\eta\leq 1$. In the following, we omit $\eta$ by assuming that file sizes are virtually scaled by a factor $\frac{1}{\eta}$. The actual value of $\eta$ will be investigated in Section \ref{sec:chunks}.

\vspace{-.1cm}
\section{Study of the Basic Model}
\vspace{-.1cm}
\label{sec:theory}

In this section, we give some theoretical results for the basic model when $D$ is a subdomain of
the Euclidean plane (or a two dimensional torus). 
We only give here the key ideas that explain the results. Detailed proofs are available in \cite{baccelli:inria-00615523}.

\vspace{-.1cm}
\subsection{Steady State}
\vspace{-.1cm}

The system's dynamics belongs to the class of spatial birth and death
processes \cite{preston75}. The births are the peer arrivals described above. The death
rate of a peer at $x$ is $\mu(x)/F$ with $\mu(x)$ given by formula (\ref{eqdeathrate}).
The first result is about the stability of the system:
\begin{prop}
\label{lem0}
If the domain $D$ in which the peers live is compact, then 
the spatial birth and death process (i.e. the positions of peers present at time $t$) forms a Markov process which is
ergodic for any birth rate $\lambda>0$.
\end{prop}

The proof of Proposition \ref{lem0} is based on a domination argument. The claim also holds in $\R^2$ but requires a more sophisticated proof that will appear in a forthcoming paper. 

According to Proposition \ref{lem0},
the model admits a steady state regime where
the peers (in the basic model all leechers) form a
stationary and ergodic point process in $D$ \cite{DalVJon:88}.

We denote by $\beta_o$ the density of the peer (leecher) point process, 
by $\mu_o$ the mean rate of a typical peer, by $W_o$ the mean latency
of a typical peer, and by $N_o$ the mean number of peers in a ball of radius $R$ around a typical peer, all in the steady state regime of the P2P dynamics.

In the following, we will also consider several approximations of the main model:
\begin{itemize}
	\item a {\em fluid regime/limit}, where the corresponding quantities will be denoted
by a subscript $f$ (e.g. $\beta_f$);
\item a heuristic description with a hat notation (e.g. $\hat{\beta}_0$)
\end{itemize}

In any of these regimes, Little's law tells us that the average density verifies $\beta=\lambda W$. 

\vspace{-.1cm}
\subsection{Fluid Limit}
\vspace{-.1cm}
\label{subsec:simplefluid}

The fluid limit consists in assuming that, in the steady state regime,
peers are distributed according to an homogeneous Poisson point process
in $D$ such that the mean number of neighbors
of any peer is large. In particular, in the fluid limit, the presence of a single peer at a given point does not impact the
distribution of the other peers.

From Campbell's formula \cite{DalVJon:88},
the mean total rate of a typical location 
of space (or of a newcomer peer) is then
\vspace{-.15cm}
\begin{equation}
\label{meanrate}
\mu_f= \beta_f 2 \pi \int_{r= 0}^R (C/r) r dr=
\beta_f 2\pi C R.
\end{equation}
Now, the fluid limit assumes that a peer sees $\mu_f$ during its whole lifetime.
We get that the mean latency of a peer is 
\vspace{-.15cm}
\begin{equation}
W_f= \frac F{\mu_f}\text{.}
\label{eq:wffmuf}
\vspace{-.15cm}
\end{equation}
Using Little's law, one gets
\vspace{-.15cm}
\begin{equation}
\label{eq:equil}
\beta_f \mu_f= \lambda F\text{.}
\vspace{-.15cm}
\end{equation} 
From \eqref{meanrate}, \eqref{eq:wffmuf} and (\ref{eq:equil}), we have
\begin{equation}\beta_f= \sqrt{
\frac{\lambda F}{2 \pi CR}}, \ 
\mu_f = \sqrt{
\lambda F 2 \pi CR}, \ 
W_f = \sqrt{ \frac F{
\lambda 2 \pi CR}}.
\label{eq:cassimple}
\end{equation}

As we see in the expression
for the mean latency in \eqref{eq:cassimple}, the fluid limit exhibits 
the same super-scalability as the toy example: in spite
of the fact that the interactions are limited in range and depend on the distance, the mean latency decreases in $\frac 1 {\sqrt{\lambda}}$ when $\lambda$ tends to
infinity and everything else is fixed. 

Note that in the fluid limit, the mean number of peers in a ball of radius $R$ around a typical peer is
\vspace{-.2cm}
\begin{equation}
\label{eq:NN}
N_f = \pi R^2 \beta_f=\sqrt{\frac \pi 2} \sqrt{\frac{\lambda F R^3}{C}}\text{.}
\vspace{-.15cm}
\end{equation}

\vspace{-.15cm}
\subsection{Dimensional Analysis}

At this point of the paper, the fluid limit is a thought experiment,
not necessarily related to the actual model. Dimensional 
analysis \cite{pi} helps to connect the two.

In the basic model, the system has 4 parameters (the range $R$, the file size $F$, the peer arrival rate
$\lambda$ and the rate parameter $C$) expressed in 3 basic physical units (meters, bits, seconds). The $\pi$-theorem~\cite{pi} allows us to strip the problem from all its parameters but one. The idea is that the behavior of a system is not affected by the physical units used to measure it. By using proper unit changes \cite{baccelli:inria-00615523}, the system can be described by just one dimensionless parameter
\vspace{-.15cm}
\begin{equation}
\rho
= \frac{\lambda F R^3}{C}.
\label{eq:theorempi}
\vspace{-.1cm}
\end{equation}

The $\pi$-theorem leaves some freedom in the choice of the parameter. By noticing that $N_f=\sqrt{\frac \pi 2} \sqrt {\rho}$, we can use $N_f$, which has a physical interpretation (the number of neighbors predicted by the fluid limit), instead of $\rho$.

The $\pi$-theorem tells us that all systems that share the same
parameter $N_f$ are similar. Now consider the union of two independent systems that use the same parameters ($\lambda$, $F$, $C$, $R$): the real model, with latency $W_o$, and the fluid model, with latency $W_f$. The ratio $\frac{W_o}{W_f}$ is a dimensionless property of the overall system, therefore it is a function of $N_f$ only. 
In other words,
there exists a dimensionless function $M(N_f)$ such that:
\vspace{-.15cm}
\begin{equation}
W_o=M(N_f)W_f\text{.}
\label{eq:Wfull}
\vspace{-.15cm}
\end{equation}

From Little's law, we also deduce the density:
\vspace{-.15cm}
\begin{equation} 
\beta_o= \beta_f M(N_f)\text{.}
\label{eq:nonameinmind}
\vspace{-.15cm}
\end{equation}

Note that the dimensional reasoning made on the basic model can be extended to other models, for instance with different rate functions or connectivity rules. Equation \eqref{eq:nonameinmind} will remain true, although the shape of $M$ may change; in particular, if the system is described by more than 4 parameters, $M$ may depend on more than one variable.

To summarize, although the system in the basic model may be subject to complex interactions and is defined by four independent parameters, dimensional analysis allows one to express its general behavior through a one-parameter function
$M$ (unknown at this point), which expresses how far the actual system is from its
fluid limit.

\vspace{-.15cm}
\subsection{Fluid as a Bound}
\vspace{-.1cm}
We now give a better understanding of the behavior of the real system through the following theorems.

\begin{thm}[Fluid as a bound]\label{conjecture1}
$M\geq 1$. In other words, the fluid regime is actually a lower bound for the mean latency and the peer density.
\end{thm}

The proof comes from a stochastic intensity argument. This property stems from the fact that as a peer uploads content to its neighbors, it makes them leave the system faster than if it did not upload anything. This is called a \emph{repulsion} effect. As a result, the mean download rate experienced by a typical peer (Palm distribution) is less than the mean download rate that would experience a virtual, non uploading, peer located at a typical location of $D$. Details can be found in \cite{baccelli:inria-00615523}.

\begin{thm}[Fluid as a limit]
\label{conjecture2}
When $N_f$ goes to infinity, $M$ goes to $1$, and the law of a typical peer latency converges weakly to an exponential random variable with parameter $1/W_f$.
\end{thm}

Theorem \ref{conjecture2} says that the fluid bound is tight:  when the number of neighbors predicted by the fluid limit tends towards infinity, the system behaves like its fluid limit.

The idea of the proof is that, when $N_f$ tends to infinity: (i) the traffic is high enough for the impact of one given peer, and thus the repulsion effect, to be neglected; (ii) the peers stay long enough to make the fluctuations slow and weak.
The fact that the rate at any point is constant in the limit implies that the latency is
exponential in the limit.

\vspace{-.15cm}
\subsection{Heuristic} 
\vspace{-.1cm}
\label{secheur}
For arbitrary  values of $N_f$, we propose to approximate $M$ by $\hat{M}$, the unique solution in $[1,\infty)$ of
\vspace{-.15cm}
\begin{equation}
\hat{M}^2\left(1-\frac{\hat M}{2N_f}\ln \left(1+\frac{2N_f}{\hat M}\right)\right)=1\text{.}
\label{eq:Mheuristic}
\vspace{-.1cm}
\end{equation}

In order to derive \eqref{eq:Mheuristic}, we use a heuristic factorization of the factorial
moment measure of order 3 of the stationary peer point process (see \cite{DalVJon:88} for
the definition of these measures) which is described in \cite{baccelli:inria-00615523}.
Informally, the method consists in computing an approximation
$\hat{u}_o$ of the  average rate of a peer assuming that: (i)  a
neighbor at distance $r$ from that peer  ``sees'' a rate
$\hat{u}_o+\frac{C}{r}$; (ii) in return, the peer ``sees'' at distance
$r$ a density of neighbors $\frac{\lambda F}{\hat{u}_o+\frac{C}{r}}$
(using \eqref{eq:equil}).

This heuristic is in line with Theorems \ref{conjecture1} and \ref{conjecture2}.

\begin{rem}
\label{rem:hardcore} \emph{When $N_f$ goes to $0$, the system admits another limit, called \emph{hard-core}, which was not presented here due to its lack of interest for real P2P systems. Nevertheless, the heuristic is in line with the hard-core limit too, which predicts that $M$ behaves like $\frac{1}{N_f}$ when $N_f$ goes to $0$ \cite{baccelli:inria-00615523}.}
\end{rem}

\vspace{-.15cm}
\subsection{Validation}
\vspace{-.1cm}
\label{subsec:simulations}

We validated and substantiated our results by means of simulations of our model. We used a discrete time simulator to evaluate the basic model for several values of $N_f$ (see \cite{baccelli:inria-00615523} for details). Key results are displayed in Figure 
\ref{fig:nm}, which allows us to check almost all results of this section in one look:\\
\noindent\textbullet~$M=1$ is a lower bound of the actual system (Theorem \ref{conjecture1});\\
\noindent\textbullet~as $N_f$ goes to $\infty$, the bound becomes tight (Theorem \ref{conjecture2});\\
\noindent\textbullet~the heuristic \eqref{eq:Mheuristic} gives a
good approximation of $M$;\\
\noindent\textbullet~as $N_f$ goes to $0$, the system behavior converges towards the hard-core limit $M=\frac{1}{N_f}$ (cf. Remark 2).

We also checked that for $N_f$ big enough, it is quite difficult to distinguish the system from a spatial birth and death process with birth parameter $\lambda$ and death parameter $1/W_f$, namely a Poisson point process of intensity $\beta_f$ (cf. \cite{baccelli:inria-00615523}).

\begin{figure}%
\begin{center}
\includegraphics[width=.85\columnwidth]{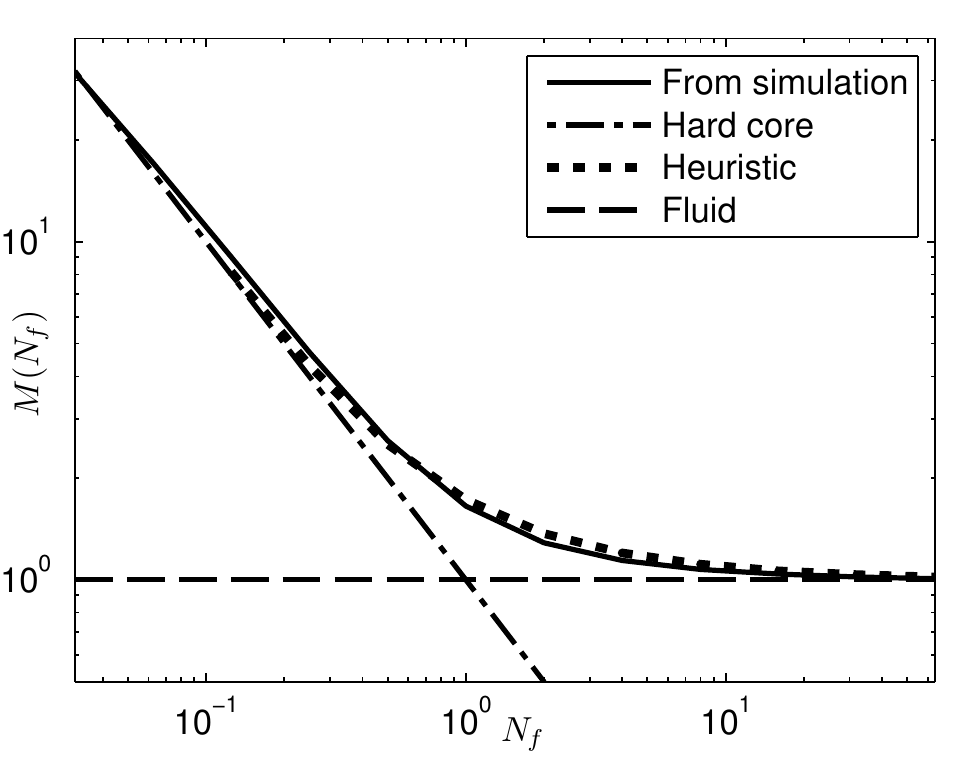}%
\vspace{-.35cm}
\caption{$M(N_f)$ in the basic model.}%
\label{fig:nm}%
\end{center}
\vspace{-.6cm}
\end{figure}


\section{Network Capacity Constraints}
\label{sec:lim}

Super-scalability naturally rises the question of the burden on the underlying network. The aim of this section is to determine the capacity required for the network elements in order to achieve the super-scalable regime identified above.

So far, the only assumptions on the network were that 1) the access is not the (only) bottleneck; 2) the network is a bottleneck, resulting into a
transfer rate between peers that depends on their distance.

This section introduces an abstract network model
on which the P2P traffic will be mapped through 
some natural shortest path routing mechanism. We
determine the mean {\em flow} that traverses a typical network element.
This flow of course depends on the protocols used in the network
which in turn determine the bit rate function.

For simplicity, we consider the fluid limit of the system.

\vspace{-.15cm}
\subsection{Network Model}
\vspace{-.1cm}

We consider an underlying network made of routers and links between them
where
\begin{compactitem}
\item routers form a realization of a spatial Poisson point process of intensity $\theta$;
\item links are the Delaunay edges (see e.g. \cite{FnT1}, Chapt. 4)
on this point process;
\item the capacity of a link is $E$;
\item each peer is directly connected to the closest router
and the path between two routers is a minimal path (with respect to
hop count) on the Delaunay graph.
\end{compactitem}

In this case, the number of links between two peers is
asymptotically proportional to the distance between them \cite{FnT1}.

Consider a straight line of the plane of length $l$. The average number of links that go through the line is $2l\sqrt{\theta}$, so the maximal traffic that can cross the line is $2El\sqrt{\theta}$. In other words, $\Xi:=2\sqrt{\theta} E$ is a parameter that describes the capacity of the network, expressed in $bits\cdot s^{-1}\cdot m^{-1}$.

\vspace{-.15cm}
\subsection{Flow Equations}
\vspace{-.1cm}

Let $\Psi(\varepsilon)$ denote the mean value of the 
P2P traffic that goes through a segment $S$ of length $\varepsilon$ in the fluid regime. By isotropy, we can focus on 
$S=[(0,-\frac \varepsilon 2), (0,\frac \varepsilon 2)]$.

A simple stochastic geometry argument shows that
\vspace{-.1cm}
\begin{equation}
\Psi=\Psi(1)= 4 \beta_f^2 \int_{0}^R r^2 f(r)  dr
\label{eq:psip2p}
\end{equation}
(see \cite{baccelli:inria-00615523}).
Using the fluid expression of the density
\vspace{-.15cm}
$$\beta_f= \sqrt{\frac{\lambda F}{2\pi\int_{0}^R r f(r) dr}}\text{,}$$
we get the key relation
\vspace{-.15cm}
\begin{equation}
\label{eq:flowf}
\Psi=\Psi(1)= \frac 2 \pi \lambda F\frac {\int_{0}^R r^2 f(r)  dr}
{\int_{0}^R r f(r) dr}.
\end{equation}
Equation \eqref{eq:flowf} holds for an arbitrary rate function $f$. For $f(r)= \frac C r$,
we get 
\vspace{-.2cm}
\begin{equation}
\label{eq:flowtcp}
\Psi=
2C \beta^2 \varepsilon R^2=
\frac 1 \pi \lambda F R.
\end{equation}

\vspace{-.2cm}
\subsection{Feasibility Condition}
\vspace{-.1cm}

Now, in order to simplify the evaluation of the P2P load on the underlying network,
we assume that (a) $\theta$ is large enough so that the hop-count
between two peers can be seen as proportional to their distance
and the flow between them as a straight line;
(b) Any rate smaller than $\Xi l$ can be transported through a
segment of length $l$. Under these assumptions, the condition for the network to sustain the rate generated by our model is 
 \vspace{-.25cm}
 \begin{equation}
\Psi < \Xi.
\label{eq:psixi}
\end{equation}

Note that the flow $\Psi$ in \eqref{eq:flowtcp} does not depend on $C$, so that condition \eqref{eq:psixi} does not either. This surprising result means that in the fluid limit, we can arbitrarily scale the individual rate of connections (thus decreasing the latency) without changing the burden on the underlying network. Of course, there is a flaw in that reasoning: increasing $C$ eventually impairs the validity of the fluid limit. As $C$ increases, $N_f$ gets smaller so we tend to leave the fluid limit and the approximations we used do not apply anymore~\cite{baccelli:inria-00615523}.

\section{Adding Chunks to the Model}
\label{sec:chunks}

This section contains a mathematical model and a simulation study
allowing one to quantify the impact of chunk availability.
An important result is that when both the number of chunks and
the parameter $N_f$ (introduced in Section \ref{sec:theory}) are
large, then the systems behaves as the chunkless fluid  
model of Section \ref{sec:theory}.

\vspace{-.15cm}
\subsection{Chunk Modeling}
\vspace{-.1cm}

We assume now that the file has a constant size $F$ and is divided into $K$ chunks of equal length. At any time, a peer is characterized by its \emph{collection}, which is the subset of chunks it fully possesses.
With respect to dimensional analysis, the system is now described by two parameters: $N_f$ and $K$.

For simplicity, we focus on the steady state taken in its fluid limit with respect to the peers, and we assume that the chunk scheduling policy is based on the following principles:

\textbullet~\emph{rarest chunk first}: when a peer can choose between chunks to download, it selects the one with fewest copies in its neighborhood; as in \cite{qiusrikant04}, we assume that this prevents the missing chunk syndrome and ensures that a peer with $k$ chunks has a collection of chunks which is independent of that of the other peers and uniform on the subsets of cardinality $k$ of the set $\{1,\ldots,K\}$;

\textbullet~ \emph{random peer order}: when it can download a given chunk from many neighbors within its range, a peer chooses one at random 
(the scheduling is not \emph{network-aware}).

There are two main ways to manage the download of simultaneous chunks: in the \emph{one-to-one} model, a peer gets one chunk from a single neighbor, while in the \emph{many-to-one} model, it can aggregate the resources of all neighbors that possess that chunk. The many-to-one approach gives better theoretical performance, as we will see below, but it requires a tight synchronization between peers that collaborate for a chunk, and thus may require an additional overhead in practice.

\vspace{-.15cm}
\subsection{Performance Study}
\vspace{-.1cm}

 An exhaustive study would require to consider the $2^F-1$ possible collections (although seeders are initially needed to bootstrap the system, we still consider a steady state with no seeder, so there is no full collection).
With the proposed assumptions, the impact of chunks mainly depends on the number of chunks already possessed by the peers. We say that a peer belongs to class $k$, for $0\leq k \leq K-1$, if it possesses exactly $k$ complete chunks. 
The following theorem gives the performance of each class in the fluid regime (by fluid regime, we mean 
i) a chunk regime where the independence and uniformity assumptions described above on the distribution of the chunks hold
and ii) a peer regime where the Poisson assumptions described in the preceding section hold).

\begin{thm}
\label{thm:chunks}
In the fluid limit, the mean total download rate of a peer of class $k$, $0\leq k \leq K$, is
\vspace{-.25cm}
\begin{equation}
\mu_k=\eta_k \mu_f\text{,}
\label{eq:mueta}
\vspace{-.15cm}
\end{equation}
\noindent where $\mu_f$ is given by \eqref{eq:cassimple}. Equation \eqref{eq:etafinal} gives the $\eta_k$'s for the many-to-one scheduling while  \eqref{eq:etafinalo} gives a lower bound for the one-to-one case.
\end{thm}

\begin{proof}
In view of our assumptions on the scheduling and on the distribution of peers,  the average rate of a given transfer is just the average over the range, that is
$\frac{1}{\pi R^2}\int_0^R2\pi r(C/r)dr= \frac{2C}{R}\text{.}$

Now, we consider a peer $p$ of class $k$ with a neighbor $q$ of class $j$. In view of our assumptions on the distribution of chunks, the probability that $q$ has at least one chunk that $p$ wants, which coincides with the probability that the set of
chunks of $q$ is not included in that of $p$, is
\vspace{-.15cm}
\begin{equation}
z(k,j)=1-{\binom k j }/{\binom K j}\text{,}
\label{eq:etakj1}
\vspace{-.15cm}
\end{equation}
with the convention that ${\binom k j }=0$ for $j>k$.
Thus, if $\beta_j$ denotes the density of class $j$, the number of neighbors from whom a given peer of class $k$ may download one chunk is
\vspace{-.15cm}
\begin{equation}
N_c=\pi R^2 \sum_{j=0}^{K-1}\beta_j z(k,j)\text{.}
\label{eq:nc}
\vspace{-.15cm}
\end{equation}

In the many-to-one model, we deduce that the average download is 
\vspace{-.5cm}
\begin{equation}
\mu_k=\frac{2C}{R}\pi R^2 \sum_{j=0}^{K-1}\beta_j z(k,j)\text{.}
\label{eq:muk}
\vspace{-.15cm}
\end{equation}

We notice then that for class $k$, \eqref{eq:equil} becomes $\beta_k=\frac{\lambda F}{K\mu_k}$.
To conclude, we define $\eta_k:=\frac{\mu_k}{\mu_f}$, where $\mu_f$ is given by \eqref{eq:cassimple}. If we replace $\beta_k$ by $\frac{\lambda F}{K\mu_k}$  in \eqref{eq:muk} and use the relationships from \eqref{eq:cassimple} and \eqref{eq:NN}, we get
\vspace{-.5cm}
\begin{equation}
\eta_k=\frac{1}{K}\sum_{j=0}^{K-1}\frac{z(k,j)}{\eta_j}\text{.}
\label{eq:etafinal}
\vspace{-.2cm}
\end{equation}

In the one-to-one model, a peer cannot download
a chunk from more than one peer. In the worst case where each of the $N_c$ peers has at most one of the desired chunks, the probability that $p$ can download any given desired chunk is $1-(1-\frac{1}{K-k})^{N_c}$, so that the average number of chunks downloaded is 
\vspace{-.15cm}
\begin{equation}
(K-k)\left(1-(1-\frac{1}{K-k})^{N_c}\right)\text{.}
\label{eq:probchunk}
\vspace{-.15cm}
\end{equation}

Adapting \eqref{eq:muk}, using the same variable changes as for the many-to-one case, and using $N_f$ as a lower bound for $N_c$, one gets:
\vspace{-.35cm}
\begin{equation}
\eta_k\geq\frac{K-k}{N_f}\left(1-(1-\frac{1}{K-k})^{N_f}\right)\text{.}
\label{eq:etafinalo}
\vspace{-.35cm}
\end{equation}
\end{proof}

Equation \eqref{eq:etafinal} is easily solved using fixed-point iterations. Notice that the computation depends solely on $K$ in the many-to-one model and on $K$ and $N_f$ in the one-to-one model.
If $\eta$ denotes the harmonic mean of the $\eta_k$'s, we verify that the overall latency $W$ is $\frac{W_f}{\eta}$. Therefore, as for the model proposed in \cite{qiusrikant04}, $\eta$ can be used to scale the results of the basic model and ignore the underlying, possibly complex, chunk exchange mechanisms.

\begin{rem}
\emph{In the basic model we had $W=M(N_f)W_f$, so we can interpret $\frac{1}{\eta}$ as $M(N_f,K)$ in the case $N_f\gg 1$.}
\end{rem}

We now study the behavior of $\eta$ in the fluid limit.

\begin{thm}
\label{thm:etaapprox}
In the many-to-one model, and in the one-to-one if $N_f$ is large enough yet fixed, we have
\vspace{-.2cm}
\begin{equation}
\eta \xrightarrow[K\rightarrow \infty]{}  1.
\label{eq:etalimit}
\vspace{-.15cm}
\end{equation}
\end{thm}
\begin{proof} [Sketch of Proof]
For the many-to-one model, we use a scaling technique that consists in letting $K$ go to infinity
so as to make the $\eta_k$ converge toward a continuous function in $[0,1)$.
The basic ingredient is the fact that the function $z$ defined in (\ref{eq:etakj1})
converges pointwise to $1$ under this scaling. The scaling of \eqref{eq:etafinal} is
\vspace{-.3cm}
\begin{equation}
\eta(x)=\int_0^{1}\frac{1}{\eta(y)}dy\text{.}
\vspace{-.1cm}
\end{equation}
It is not difficult to show that $\eta=1$ is the unique positive solution
solution of this functional equation,
which proves \eqref{eq:etalimit} for the many-to-one case.

In the one-to-one model, \eqref{eq:etalimit} is straightforward when noticing that $\eta$ is always smaller than or equal to $1$ (the overall download capacity is lowered because of availability issues). The limit of \eqref{eq:etafinalo} when $K$ tends to $\infty$ allows one to conclude.
\end{proof}

\begin{figure}%
\centering
\includegraphics[width=.8\columnwidth]{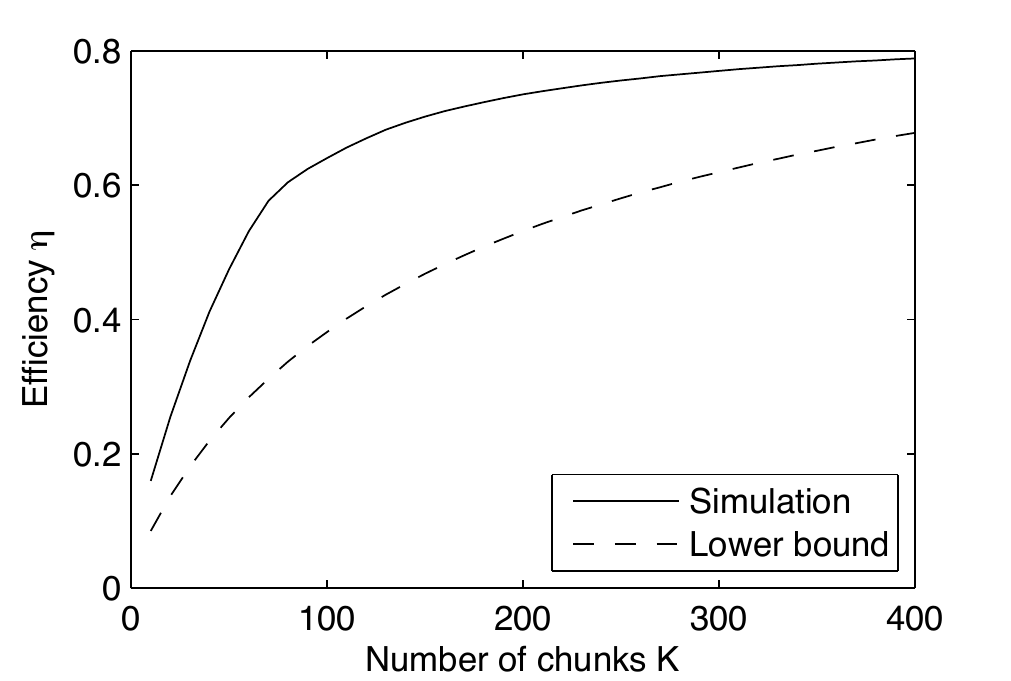}%
\vspace{-.25cm}
\caption{Efficiency $\eta$ as a function of $K$ ($N_f=40$). \label{fig:eta}}%
\vspace{-.5cm}
\end{figure}

The fact that a peer cannot upload a given chunk from more than one peer badly impacts the performance of the one-to-one model, compared to many-to-one. This is especially true at the end of the download, when a peer may have more useful neighbors than remaining chunks.
This fact was empirically observed by Bram Cohen in his original BitTorrent design, where he proposed to use one-to-one (which is easier to maintain) most of the time except for the very few last chunks, where peers switch to many-to-one (endgame behavior \cite{bittorrent}).

\subsection{Validation}
\vspace{-.1cm}

We simulate the system with chunks in order to substantiate our claims, using a simple rarest first chunk selection and random peer selection like the one proposed. Synchronization is one-to-one.

First, we validate the assumption on the distribution of chunks by checking the impact of the presence of a chunk at some peer on the presence of this chunk at the neighboring peers. For instance, for $N_f=40$, $K=200$,
we verified that a peer sees in average $29.22$ copies of a chunk it possesses (itself not included), and 
 $29.10$ copies of a chunk it misses. This and more detailed correlation analysis (that cannot be included here due to space limitation) are quite conclusive.

We launched many trials to verify our results.  Figure \ref{fig:eta} displays the value of $\eta$ for several values of $K$. One verifies that the system has a better performance than the proposed lower bound, and the right behavior when $K$ grows.

\vspace{-.15cm}
\subsection{Conclusion on Chunks}
\vspace{-.1cm}

We showed (through analysis and simulation) that in the fluid limit ($N_f\gg 1$), when $K\gg 1$, the system with chunks behaves like the fluid chunkless model of Section \ref{sec:theory} with an appropriate efficiency parameter $\eta$, which we described.

The parameter $\eta$ can be close to $1$ if $K$ is large enough, with $N_f$ being fixed in the one-to-one model. In this last case, super-scalability could be impacted: as $\lambda$ increases, so does $N_f$ and if $K$ is fixed, the lower bound converges to $0$ (simulations confirm that this is also the case for $\eta$). The possible workarounds for this issue are: to use many-to-one, or equivalently one-to-one with endgame, to get rid of the last chunks bottleneck; to limit the number of neighbors in order to keep $N_f$ bounded (this will be detailed in Section \ref{subsec:limiteddegree}).

\vspace{-.1cm}
\section{Extensions of the Basic Model}\label{sec:extend}
\vspace{-.1cm}

The aim of this section is to show that our analysis can be extended in several ways
and take important practical phenomena into account.
Unless otherwise stated, we will place ourselves in the fluid regime, but the dimensional analysis approach can be used with all extensions to relate the fluid limit to the real system through some function $M$. As we have seen when introducing the chunks, if an extension introduces new parameters, $M$ can be a function of several dimensionless variables (replacing $N_f$).

For sake of clarity, the proposed extensions are presented separately, but interleaving extensions is straightforward in the fluid limit. Outside the fluid limit, the complexity of mixed extensions will mainly depend on the complexity of the corresponding $M$ function.

\vspace{-.2cm}
\subsection{More General Rate Functions}
\vspace{-.1cm}
\label{sec:generalrate}

\begin{table}
\begin{threeparttable}
\caption{Some rate functions with explicit strength $\gamma$ \label{tab:gamma}}
\begin{tabular}{|c|p{2.4cm}|c|}
\hline
$f(r)$ & Interpretation & $\gamma=2\pi\int_0^R rf(r)dr$ \\
\hline
$\frac{C}{r}$ & TCP-like & $2\pi CR$ \\
\hline
$U$ & UDP-like (constant)  & $\pi UR^2$ \\
\hline
$\frac Cr \wedge U$ & TCP with per-flow limitation & $
\pi \left(2CR-\frac{C^2} U\right)
$\tnote{a}  \\
\hline
$\frac{C}{r+q}$ & TCP with offset & $2\pi C\left(R-q\ln(1+\frac{R}{q})\right)$ \\
\hline
$\frac{C}{r}-o$ & TCP with overhead & $\pi R \left(2C-oR\right)$\tnote{b} \\
\hline
$\frac 1 2 \ln\left(1+\frac C {r^\alpha}\right)$ & SNR Wireless &
$
\frac{\pi^2 C^{\frac 2 \alpha}}{
2 \sin\left(\frac {2\pi}{\alpha}\right)}\text{ for $R=\infty$\tnote{c}}$\\
\hline
\end{tabular}
\begin{tablenotes}
\item[a] For $C\leq UR$; $C\geq UR$ is the UDP-like case.
\item[b] For $\frac{C}{R}\geq o$; otherwise replace $R$ by $\frac{C}{o}$.
\item[c] There is no closed form for$R<\infty$ in most cases. However, for $\alpha=4$, we have $\gamma = \pi \left( { {R^2}\log (1 + \frac C {R^4})
+ \sqrt{C} \arctan (\frac{R^2}{\sqrt{C}})} \right)$.
\end{tablenotes}
\end{threeparttable}
\vspace{-.6cm}
\end{table}

While we focused for the basic model on the rate function \eqref{eqr0}, all our results can easily be generalized to any rate function $f$ such that $\int_{r=0}^Rrf(r)dr<\infty$.

For a rate function $f$, the fluid rate Equation \eqref{meanrate} becomes
\vspace{-.15cm}
\begin{equation}
\mu_f=\beta_f \gamma\text{, with }\gamma=2 \pi \int_{r=0}^Rrf(r)dr\text{.}
\label{eq:generalrate}
\vspace{-.15cm}
\end{equation}
The characteristic $\gamma$, which is expressed in $bits\cdot s^{-1}\cdot m^2$, is the sum of $f$ over its range, so we call it the \emph{strength} of $f$. Once $\gamma$ is known, we can generalize \eqref{eq:cassimple} as
\vspace{-.1cm}
\begin{eqnarray}
\beta_f  =  \sqrt{ \frac{\lambda F}{\gamma}},\
\mu_f  =  \sqrt{ \lambda F \gamma},\
W_f   =  \sqrt{ \frac F{ \lambda \gamma}}.
\label{eq:generalperf}
\vspace{-.1cm}
\end{eqnarray}

We observe that the scaling in $\frac 1 {\sqrt{\lambda}}$ still holds. For the rest of the paper, we use directly the strength $\gamma$ instead of \eqref{eqr0}.

Table \ref{tab:gamma} gives the strength of the following rate functions:
\begin{itemize}
	\item The TCP-like example of the basic model;
	\item Constant rate function, where each flow has a bandwidth $U$. This corresponds for instance to the case where the transport protocol is UDP and bandwidth is limited by the application;
	\item Mix of the above, where the rate is TCP-like with an upper bound set by the application;
	\item TCP-like with some additive offset $q$ that accounts for the mean delay in the two access networks;
	\item Capacity of a wireless AWGN channel.
\end{itemize}

 In most cases, the heuristic approximation $\hat M$ can be adapted to $f$. For instance, a constant $f$ leads to (cf \cite {baccelli:inria-00615523})
\begin{equation}
\hat M=\sqrt{1+\left(\frac{1}{2N_f}\right)^2}+\frac{1}{2N_f}\text{.}
\label{eq:Mheuusoluce}
\vspace{-.1cm}
\end{equation}

If $R=\infty$, the system parameter $N_f=\pi R^2\beta_f$ is not properly defined anymore, which impairs a direct introduction of $M$. If  $\int_{r>0}r^2f(r)dr<\infty$, a simple workaround is to use the following ratio
(already considered in \eqref{eq:flowf})
\vspace{-.1cm}
\begin{equation}
\tilde{R}:=\frac {\int_{r>0} r^2 f(r)  dr}
{\int_{r>0} r f(r) dr}\vspace{-.1cm}\end{equation}
instead of $R$ and to extend the dimensional analysis accordingly
($\tilde{R}$ being interpreted as the \emph{typical} range of $f$).
If $\int_{r>0}r^2f(r)dr=\infty$, then according to \eqref{eq:psip2p} the traffic load intensity is infinite, so the rate function is probably ill-defined with respect to the underlying, capacity-limited, network.

\vspace{-.25cm}
\subsection{Permanent Servers}
\vspace{-.1cm}

The system may benefit from servers, or eternal 
seeders\footnote{This is distinct from the case where leechers can seed for some time after they complete their download, which is addressed in \ref{subsec:seeders}}. For instance they can be introduced to: (i) solve the issue of chunk availability by being able to provide any asked chunk; (ii) allow to consider hybrid systems that combine classical server solutions and a P2P approach; (iii) avoid the fact that in our model, the latency goes to $\infty$ when $\lambda$ goes to $0$ (non-popular content syndrome).

We focus on the basic model.

The servers are characterized by their density of bitrate $U_C$, expressed 
in $bit\cdot s^{-1}\cdot m^{-2}$, so that if $\beta_f$ is the peer density, 
a typical peer gets $\frac{U_C}{\beta_f}$ from the servers.

To describe the system, we need another dimensionless parameter in addition to $N_f$. We conveniently choose $\chi:=\frac{U_C}{\lambda F}$, which expresses the ratio between the density of rate needed by the system and the density of rate provided by the servers. If $\chi\geq 1$, then the permanent rate from servers is sufficient to serve the peers, otherwise P2P transfer is needed for stability.

Let us focus on the two limiting cases: the system is mainly client/server ($\chi\gg 1$), or the system is mainly P2P with a small server-assistance ($\chi\ll 1$). The case $\chi\ll 1$ can be seen as a scenario where servers are here mainly for insuring chunk availability.

If $\chi\gg 1$, then almost all resources come from the servers. This implies that the point process is hard--core (a peer sees almost no neighbor in its range while it is a leecher, otherwise the P2P traffic would not be negligible), so a peer can collect all the available bandwidth in its range. We deduce the average latency:
\vspace{-.3cm}
\begin{equation}
W_{C}\approx \frac{F}{\pi R^2 U_C}.
\label{eq:whc}
\end{equation}

For $\chi \ll 1$, in the fluid limit ($N_f\gg 1$), we can adapt \eqref{meanrate}, which gives
\vspace{-.3cm}
\begin{equation}
\mu_{f,C}=\beta_{f,C}\gamma+\frac{U_C}{\beta_{f,C}}\text{,}
\end{equation}
from which we deduce 
\vspace{-.15cm}
\begin{equation}
W_{f,C} = \sqrt{ \frac {F-\frac{U_C}{\lambda}}{
\lambda \gamma}}=W_f\sqrt{1-\chi}\approx W_f.
\label{eq:wwithservers}
\vspace{-.2cm}
\end{equation}

\vspace{-.2cm}
\subsection{Abandonment}
\vspace{-.1cm}
Here we consider the case where all leechers have some abandonment rate.
Let $a$ denote this rate. 
In the stationary state, we have $\lambda=(\frac{\mu_f}{F}+a)\beta_f$.
From \eqref{eq:generalrate}, we deduce $\mu_f^2+\mu_faF=\lambda F \gamma$. The positive solution of this equation is 
\vspace{-.1cm}
\begin{equation}
\mu_f=\sqrt{\lambda F \gamma+\left(\frac{aF}{2}\right)^2}-\frac{aF}{2}\text{.}
\label{eq:abandonment}
\end{equation}
The analysis can hence be extended without difficulties. For instance, the abandonment ratio is given by $\frac{aF}{\mu_f+aF}$.

\subsection{Per Peer Rate Limitation} 
\vspace{-.1cm}
Due to the asymmetric nature of certain access networks (e.g. ADSL), the uplink rate is often the
most important access rate limitation. Let $U$ denote (here) the average upload capacity of a peer;
then the average rate in the fluid limit should be such that
\begin{equation}
\mu_f=\sqrt{\lambda F \gamma} \le U.
\label{eq:caspeeraccess}
\end{equation}
If $\gamma=2\pi R C$ (basic model), a dimensioning rule could be to choose $R=\frac{U^2}{\lambda F 2 \pi C}$ so that all available capacity is used.
\vspace{-.1cm}
\subsection{Leechers and Seeders}
\vspace{-.1cm}
\label{subsec:seeders}
When a leecher has obtained all its chunks, it can become a seeder and remains such 
for a duration $T_S$.
In this setting, there is a density of seeders $\lambda T_S$ in the stationary regime.

In the fluid limit with seeders, \eqref{eq:generalrate} becomes
\begin{equation}
\label{meanratesee}
\mu_{f,S}=
(\beta_{f,S}+\lambda T_S) \gamma\text{.}
\end{equation}
Using \eqref{eq:equil} and $F=W_{f,S}\mu_{f,S}$, we get
\begin{equation}
W_{f,S}^2+W_{f,S}T_S=W_{f}^2\text{.}
\label{eq:wfseeder}
\end{equation}
The positive solution of this equation is
\begin{equation}
W_{f,S}=\sqrt{W_{f}^2+\left(\frac{T_S}{2}\right)^2}-\frac{T_S}{2}.
\label{eq:wfseedersoluce}
\end{equation}
In particular, we have $W_{f,S}\approx W_{f}$ for $T_S\ll W_{f}$ and  $W_{f,S}\approx \frac{W^2_{f}}{T_S}$ for $T_s\gg W_{f}$.

\begin{rem}
\emph{ Seeders can also greatly improve the performance in the case where $N_f$ is small, by ensuring that a leecher can find peers in its range with high probability (cf \cite{baccelli:inria-00615523} for more details).}
\end{rem}

\subsection{Limited Degree}
\label{subsec:limiteddegree}
In the basic model, we limit connectivity by range for mathematical tractability, but in practice, most P2P systems use a limitation based on the number of connections per peer.

However, degree limited connectivity can be linked to our model. Consider that a ALTO-like mechanism allows each peer to connect to its $L$ {\em nearest peers}. If $L$ is high enough, it will be identified to $N_f$ and the behavior will be fluid. The degree connectivity can then be approximated by a range connectivity such that $L$, $R$ and $\beta$ verify
\begin{equation}
\pi R^2 \beta =L\text{.}
\label{eq:pir2l}
\end{equation}

Using \eqref{eq:equil} and \eqref{eq:generalrate}, we get an equation that $\beta$ must verify:
\begin{equation}
\beta^2\gamma(\beta)=\lambda F\text{,}
\label{eq:adapt1}
\end{equation}
\noindent where $\gamma(\beta)$ is the strength of the rate function $f$ when using $R=\sqrt{\frac{L}{\pi\beta}}$ (see for instance Table \ref{tab:gamma}).

Once $\beta$ is known, we deduce $W=\frac{\beta}{\lambda}$. For instance, using the rate function of the basic model, one gets
\begin{equation}
W=\left(\frac{F}{2C}\right)^{\frac{2}{3}}\left(\frac{1}{\pi \lambda L}\right)^{\frac{1}{3}}\text{.}
\label{eq:TCPdegree}
\end{equation}

We observe that the super-scalability property still exists (although slightly diminished), despite the fact that each peer has a limited number of neighbors. This is a consequence of having a decreasing $f$ function: as the arrival rate increases, so does the density and thus the rate of individual connections. To compare with, a system with a constant rate function like in the toy example is simply scalable if the degree connectivity is limited (the latency is obviously $W=\frac{F}{LU}$).

Finally, we can propose a fluid model that encompasses both the range and degree models. Consider that there is a function $p(r,\beta)$ that describes the probability that a peer connects to another one given that their distance is $r$ and the density is $\beta$.

The equation to solve is still \eqref{eq:adapt1}, except that we now define
\begin{equation}
\gamma(\beta)=\int_{r>0} 2\pi r f(r) p(r,\beta) dr\text{.}
\label{eq:gammap}
\end{equation}

Under this formalism, the range model is simply $p(r,\beta)=1_{r\le R}$, while the degree limited model corresponds to $p(r,\beta)=1_{r\le \sqrt{\frac{L}{\pi\beta}}}$. For these two cases, the function $p$ corresponds to very simple overlays, but it could be used to model more complex structures like random geometric graphs.

\section{Conclusion}
\label{sec:conclusion}

In a P2P system with
a rate function $f$ and a range $R$,
the following general law quantifying
P2P super-scalability was identified:
the stationary latency is of the form
\begin{equation}
\label{eq:grail}
W_o=M\left(\sqrt{\frac{\pi^2 R^4 \lambda F}{\gamma}}\right)
\sqrt{\frac{ F}{\lambda \gamma}},
\end{equation}
with $\gamma=2\pi\int_{0}^R f(r) dr$ and
with $M(x)$ a function which is larger than 1 and tends to 1 
when $x$ tends to infinity.
In the TCP case,
the function $x\to M(x)$ is decreasing and hence reinforces 
super-scalability.

The conditions for the super-scalability formula \eqref{eq:grail}
to hold were also identified:
(1) The number of chunks should be large (so as to be in the fluid regime w.r.t. chunks);
(2) The parameter $N_f =\pi R^2 \sqrt{\lambda F/\gamma}$ should be large (so as to be in the fluid regime w.r.t. peers).
If (1) or (2) do not hold, then chunk/peer availability issues will
dominate and the model breaks down;
(3) the network should have the
capacity to cope with the P2P traffic, i.e.
\begin{equation}
E \sqrt{\theta}>
\frac{2 \lambda F}{\gamma} \int_0^R r^2 f(r) dr, 
\end{equation}
where $\theta$ is the spatial intensity of routers
and $E$ the typical link capacity. Hence the 
capacity of the network should scale like $\lambda$ if other
parameters are unchanged. If this condition does not hold, the network
cannot cope with the traffic and the model breaks down;
(4) The access should not be the bottleneck,
which translates into the requirement
\begin{equation}
U > \sqrt{\lambda F \gamma},
\end{equation}
where $U$ denotes the (total) upload capacity of each peer.
In other words, the latter should scale like $\sqrt{\lambda}$.
If this is not the case, then classical access bottleneck model
should be used. 

\paragraph*{Acknowlegments} This work has been carried out at LINCS (\url{http://www.lincs.fr}) and has been partly funded by the EIT ICT Labs Projects \emph{Fundamentals of Networking} and \emph{Distributed Content Delivery in Wireless Networks}.

\bibliographystyle{IEEEtran}

\label{LastPage}

\end{document}